\documentclass[a4paper]{amsart}
\usepackage{amsmath,amsthm,amsfonts,amssymb,amscd,amstext,mathrsfs}
\usepackage{url}
\theoremstyle{plain}
\usepackage[english]{babel}
\newcommand{\bC}{\mathbb C}

\newcommand{\bR}{\mathbb R}

\newcommand{\cA}{\mathcal A}

\newcommand{\mt}{\mapsto}
\newcommand{\ra}{\rightarrow}

\newcommand{\lra}{\longrightarrow}
\newcommand{\lmt}{\longmapsto}

\newtheorem{definition}{Definition}

\newtheorem{thm}[definition]{Theorem}

\newtheorem{prop}[definition]{Proposition}
\numberwithin{equation}{section}
\numberwithin{definition}{section}

\begin{document}
\title{Smooth structures on the field of prequantum Hilbert spaces}
\author{R\'obert Sz\H{o}ke}
\address{Department of Analysis, Institute of Mathematics,
ELTE E\"otv\"os Lor\'and University,
P\'azm\'any P\'eter s\'et\'any 1/c, Budapest 1117, Hungary,
ORCiD:0000-0002-8723-1068}


\email{rszoke@cs.elte.hu}
\thanks{This research was   supported by OTKA grant K112703}


\begin{abstract} 
When there is a family of complex structures on the phase space,
parametrized by a set $S$, the prequantum Hilbert spaces produced by geometric
quantization, using the half-form correction, also depends on these parameters.
This way we obtain a field of Hilbert spaces $p:H^{pr Q}\ra S$. We show that
this field can have natural inequivalent smooth Hilbert bundle structures.
\end{abstract}
\maketitle
\section{Introduction}\label{S:intr}

Let    $M^m$ be an
$m$--di\-men\-si\-o\-nal compact Riemannian manifold, that is
 the
configuration space of a classical mechanical system.
 Geometric quantization  intends  to construct a Hilbert space
 (the quantum Hilbert space) associated to
this system, in a natural way.

According to the recipe of Kostant and Souriau 
\cite{Ko,So,Wo},
the first step in this process,
   is to pass to phase space $(N,\omega)$,  where
$\omega$ is a symplectic form, and then to choose
a Hermitian line bundle with connection $E\to N$,
with curvature  $-i\omega$.
 The  prequantum Hilbert space $H^{pr Q}$ then
  consists of the $L^2$ sections of $E$. For physical reasons this Hilbert space  is too big
  and one needs to reduce its size.
 Hence the next step is to select a  Hilbert subspace of $H^{pr Q}$.
  Frequently this is done with the help of  a complex structure on $N$.
  If $\omega$ is a K\"ahler form, $E$ also inherits a holomorphic structure
  and
  one can take the quantum Hilbert space $H$
  to be the holomorphic $L^2$ sections of $E$. Often one  includes in this 
construction the so called half-form correction. Assume
that the   canonical
 bundle $K_N$ admits 
   a square root $\kappa$. 
Then the corrected prequantum Hilbert space  consists of 
the  $L^2$   sections of $E\otimes\kappa$ and the corrected
quantum Hilbert space
$H^{corr}$ consists of 
the  holomorphic $L^2$   sections of $E\otimes\kappa$.

In certain situations not only one but an entire family (parametrized
by a set $S$) of complex
structures exists on $N$ resulting a family of quantum Hilbert spaces $H_s$
(or $H^{corr}_s$).
The question then arises: is it possible to identify these Hilbert spaces in
a canonical way?

Suppose the set $S$ admits a  smooth manifold structure.
Axelrod et al. \cite{ADW} and  Hitchin \cite{Hi} suggest  to view the family 
     $\{H_s : s\in S\}$ (or 
$H^{corr}_s$) 
as the fibers of a   Hilbert bundle $H\ra S$ (resp. $H^{corr}\ra S$)
endowed with some Hermitian 
connection, the quantum connection.
In situations when this is true, parallel transport along a curve in
$S$ would yield a unitary map between different fibers. If furthermore the
quantum connection was flat, we would even have a (locally) path-independent
method of identification.

Now  in the uncorrected situation $H^{pr Q}$ is fixed, it does not depend
on the parameter $s$.
Since $H_s$ are all closed subspaces of this Hilbert space,
to implement the idea above, we could 
 try to view this family as the fibers of a Hilbert subbundle $H\ra S$
of the
trivial Hilbert bundle $\mathcal L:=S\times H^{pr Q}\ra S$.
In this case the quantum connection
would arise from the canonical flat hermitian connection on $\mathcal L$ by
orthogonal projection.
In very specific situations (\cite{ADW,Wo})  this works, but we are not
aware of any general theorem that would guarantee that $H\ra S$ is indeed a
Hilbert subbundle of $L$.

In the corrected version a further difficulty arises:
the prequantum Hilbert
spaces $H^{pr Q}_s$
(the $L^2$ sections of $E\otimes\kappa_s$) themselves are no longer
fixed but form a family of which the corrected quantum Hilbert spaces form
a family of Hilbert subspaces.

Thus before we
 try to show that $\{H^{corr}_s, s\in S\}$ is a ``subbundle'' of
 $\{H^{pr Q}_s, s\in S\}$,
 we need to see 
 whether  $\{H^{pr Q}_s,s\in S\}$ forms  a Hilbert bundle or not.

The main purpose of this note
 is to demonstrate that in certain cases there are at least two natural 
 inequivalent
 ways to make $\{H^{pr Q}_s\}$ a Hilbert bundle. The topology on these bundles
 is the same, but their smooth structure (and therefore the set of smooth
 sections) are different.
This is the content of Theorem~\ref{T:51}.

Therefore it is better not to try to explain the smooth structure
(and the quantum connection)
on $H^{corr}$ (as in \cite{ADW, FMMN1, FMMN2})  through $H^{pr Q}$ but rather
do it directly, using the family of complex structures defined on $N$.  
This was our motivation and reason why 
we introduced in \cite{LSz3} the notion of smooth and analytic fields
of Hilbert spaces generalizing Hilbert bundles with a hermitian connection.
 It is not clear  whether $H^{corr}\ra S$ really forms a smooth
 Hilbert subbundle or not 
in $H^{pr Q}$ equipped with either of the smooth structures considered in
Theorem~\ref{T:51}.

\section{A unitary representation of $\cA_+$}

The  affine semigroup $\mathcal A$ is the  semigroup
of affine transformations $\sigma:\bR\to\bR$, $t\mt a+bt$, where $a,b\in\bR$.
 Here $a$ and $b$ provide
 global coordinates  on $\mathcal A,$ and identify it with $\bR^2$. With this
single  coordinate chart $\cA$
 becomes a Lie semigroup. 
The subset of $\cA$  with $b>0$ forms a Lie (sub)group, that we shall
denote by $\mathcal A_+$. In the above coordinates 
 $\mathcal A_+$  is identified,  
as a smooth manifold, with the upper half plane $S$.

$\mathcal A_+$ acts naturally not only on $\bR$ but on $\bC$  
as well.
 For an $s\in S$,  let $\sigma_s\in \cA_+$ be the unique
 element with  $\sigma_s(i)=s$. The map $S\ra \mathcal A_+$, $s\mt \sigma_s$ is
 simply the inverse of the coordinate chart of $\mathcal A_+$.
 
The curves $\alpha, \beta:\bR\ra\mathcal A_+$ defined as
\begin{equation}\label{E:11}
\alpha(u)=\{t\mt u+t\},\quad
\beta(u)=\{t\mt e^ut\}
\end{equation}
are 1-parameter subgroups
in $\mathcal A_+$ and $\dot\alpha(0), \dot\beta(0)$ form a basis in its Lie
algebra.
Define the homomorphism $\chi:\mathcal A_+\ra (0,\infty)$ 
by  $\chi(\sigma)=b$ if
$\sigma(t)=a+bt$.

Let now $M^m$ be an $m-$dimensional, compact, smooth manifold equipped with a 
Riemannian metric.
 We take the {\it phase space $N$} to be the manifold of parametrized geodesics
 $x:\bR\ra M$. 
 Any $t_0\in\bR$ 
induces a diffeomorphism $N\ni x\mt\dot x(t_0)\in TM$, and  the pull back of 
the  canonical 
symplectic form of $TM\approx T^*M$ is independent of $t_0$ (cf. \cite{Wo})
that
we shall denote  by 
$\omega$.

Associating with a point $q\in M$ the constant geodesic
 $\equiv q$, identifies $M$ as a submanifold in $N$. In the $TM$ picture this
 is just the zero section of the tangent bundle.
 
Composition of geodesics with the elements of $\mathcal A$ defines a right action
of $\mathcal A$ on $N$. Induced
by this action, we get the
vector fields $\mathcal X, \mathcal Y$ on $N$, corresponding to $\dot\alpha(0)$
and $\dot\beta(0)$.  
When $N$ is replaced by $TM$, $\mathcal X$ is the vector field induced by the
geodesic flow and $\mathcal Y$ corresponds to the Euler vector field.
 For  $\sigma\in\mathcal A$, $x\in N$ let $A_\sigma:N\ra N$ be defined by 
$A_\sigma x=x\circ\sigma$. As one easily checks,
\begin{equation}\label{E:12}
A^*_\sigma\omega=\chi(\sigma)\omega.
\end{equation}
Let
$\mathcal E=\frac{\omega^m}{m!}$ be the Liouville volume form on $N$.
From \eqref{E:12} we get
\begin{equation}\label{E:13}
A^*_\sigma(\mathcal E)=\chi(\sigma)^m\mathcal E.
\end{equation}
Denote by $L^2(N,\mathcal E)$ the complex valued square integrable
(w.r.t. the volume form $\mathcal E$) functions
on $N$ and $U(L^2(N,\mathcal E))$ the unitary self maps of this  Hilbert space.
\begin{thm}\label{T:21} With $\sigma\in\mathcal A_+$, the map
\begin{equation}\label{E:14}
\begin{matrix} \rho(\sigma):L^2(N,\mathcal E)&\ra&L^2(N,\mathcal E)\\
f&\mt&\chi(\sigma)^{\frac{m}{2}}f\circ A_\sigma
\end{matrix},
\end{equation}
is unitary and yields a unitary representation
\[
\rho:\mathcal A_+\ra U(L^2(N,\mathcal E)).
\]
The map
\begin{equation}\label{E:15}
\varrho:\mathcal A_+\times L^2(N,\mathcal E)\longrightarrow L^2(N,\mathcal E),
\end{equation}
defined by $\varrho(\sigma,f):=\rho(\sigma)f$ is continuous, but not
differentiable.
\end{thm}
\begin{proof} \eqref{E:13} implies

\[
\int\limits_N|\rho(\sigma)f|^2\mathcal E=
\int\limits_N\chi(\sigma)^m|f\circ A_\sigma|^2\mathcal E=
\int\limits_N A^*_\sigma(|f|^2\mathcal E)=
\int\limits_N|f|^2\mathcal E
\]
hence $\rho(\sigma)$ in \eqref{E:14} is  unitary. Also
if $\sigma,\sigma'\in\mathcal A_+$,
\[
\rho(\sigma\sigma')f=\chi(\sigma\sigma')^{\frac{m}{2}}f\circ A_{\sigma\sigma'}=
\]
\[
\chi(\sigma\sigma')^{\frac{m}{2}}f\circ(A_{\sigma'}\circ
A_\sigma)=\chi(\sigma)^{\frac{m}{2}}(\rho(\sigma')f)\circ A_\sigma=
\rho(\sigma)(\rho(\sigma')f)
\]
shows that $\rho$ is a representation.

We prove the continuity of $\varrho$ in two steps.

{\it First step:}  Let $g\in L^2(N,\mathcal E)$ and  $\sigma\in\mathcal A_+$ be fixed.
Assume that $g$ is  continuous 
with compact support.
We want to show that, when  $f\in L^2(N,\mathcal E)$ is close to $g$
and $\sigma'\in\mathcal A_+$ is close to $\sigma'$, then $\rho(\sigma')f$  is
 close to
$\rho(\sigma)g$.
 In the rest of the proof
 norm always refers to the $L^2$ norm w.r.t. the volume
 form $\mathcal E$.
\[
\|\rho(\sigma')f-\rho(\sigma)g\|\le\|\rho(\sigma')f-\rho(\sigma')g\|+
\|\rho(\sigma')g-\rho(\sigma)g\|=\|f-g\|+\|\rho(\sigma')g-\rho(\sigma)g\|,
\]
since $\rho(\sigma')$ is unitary. 
\[
\begin{aligned}
\|\rho(\sigma')g-\rho(\sigma)g\|^2&=\int\limits_N|\chi(\sigma')^{\frac{m}{2}}
g\circ A_{\sigma'}-\chi(\sigma)^{\frac{m}{2}}
g\circ A_{\sigma}|^2\mathcal E=\\
&=\int\limits_N\chi(\sigma)^m|(\chi(\sigma'
\sigma^{-1}))^{\frac{m}{2}}g\circ A_{\sigma^{-1}\sigma'}-g|^2
\circ A_\sigma\mathcal E=\\
&=\int\limits_N|(\chi(\sigma'
\sigma^{-1}))^{\frac{m}{2}}g\circ A_{\sigma^{-1}\sigma'}-g|^2\mathcal E,
\end{aligned}
\]
because $\rho(\sigma)$ is also unitary. 
This implies
\[
\|\rho(\sigma')g-\rho(\sigma)g\|\le
\sqrt{\int\limits_N((\chi(\sigma'\sigma^{-1}))^{\frac{m}{2}}-1)^2
|g\circ A_{\sigma^{-1}\sigma'}|^2\mathcal E}+
\sqrt{\int\limits_N|g\circ A_{\sigma^{-1}\sigma'}-g|^2\mathcal E}
\]
Denote the two terms on the right hand side
 by $I$ and $II$. Using the unitarity of
$\rho(\sigma^{-1}\sigma')$ we get
\[
I=\chi(\sigma{\sigma'}^{-1})^{\frac{m}{2}}
\left|\chi(\sigma'\sigma^{-1})^{\frac{m}{2}}-1\right|\|g\|
\]
If $\sigma'$ is close  to $\sigma$, $\chi(\sigma'\sigma^{-1})$
is near to $1$ and so $I$ is close to zero.
Also  in this case $A_{\sigma^{-1}\sigma'}$ is close to the
identity diffeomorphism of $N$.
Because of our choice $g$  is
uniformly continuous on its support.
Consequently $II$ is also near to zero.
All these imply the continuity of $\Lambda$
in $(\sigma,g)$.

{\it Second step:} Let now $g\in L^2(N,\mathcal E)$ be arbitrary. Choose a $g_1$
near $g$ that is continuous with compact support. Let $f\in L^2(N,\mathcal E)$.
Then
\[
\|\rho(\sigma')f-\rho(\sigma)g\|\le \|\rho(\sigma')f-\rho(\sigma')g\|+
\|\rho(\sigma')g-\rho(\sigma')g_1\|+
\]
\[
+\|\rho(\sigma')g_1-\rho(\sigma)g_1\|+
\|\rho(\sigma)g_1-\rho(\sigma)g\|
\]
\[
=\|f-g\|+\|g-g_1\|+\|\rho(\sigma')g_1-\rho(\sigma)g_1\|+\|g_1-g\|
\]
and applying the first step to $g_1$ we get the continuity of $\varrho$
in $(\sigma,g)$.

To justify that  $\varrho$ is not differentiable, it suffices to show that
its partial derivative with respect to the 
$\sigma$ variable
    does not exist.
    
Consider the 1-parameter subgroups $\alpha,\beta$ in $\mathcal A_+$ from \eqref{E:11}
 and let $g\in L^2(N,\mathcal E)$ be arbitrary.
 Then
 $\varrho(\alpha(s),g)=g\circ A_{\alpha(s)}$,
 $\varrho(\beta(s),g)=e^{\frac{sm}{2}}g\circ A_{\beta(s)}$.
So we would get
\[
\left.\frac{d}{ds}\right|_{s=0}\varrho(\alpha(s),g)=\mathcal Xg,
\quad \left.\frac{d}{ds}\right|_{s=0}\varrho(\beta(s),g)=\frac{m}{2}g+\mathcal Yg.
\]
Of course for a generic $g$ neither $\mathcal Xg$, nor $\mathcal Yg$ exists and even if
it does it is not necessarily square integrable.
\qed
\end{proof}

\section{Adapted complex structures}\label{S:acs}

First  we recall  some important facts
on adapted complex structures mainly from \cite{LSz2}.
If we are given a complex manifold structure on $\cA_+$, 
a complex structure on $N$ 
is called {\it adapted} if for every $x\in N$ the orbit map 
$\cA_+\ni\sigma\mt x\sigma
\in N$
is holomorphic.
An adapted complex structure on $N$ can exist only if
the initial compex structure on $\cA$ is left invariant. The left invariant 
complex structures on $\cA$ are parametrized by the points of 
$\bC\setminus\bR$.
For each $s\in \bC\setminus\bR$ and corresponding left
 invariant complex structure $I(s)$ on $\cA$, if an $I(s)$ adapted complex 
structure $J(s)$ exists on $N$, 
then this structure is unique and if $J(i)$ exists, then $J(s)$ also exists
 for all $s$ in $s\in \bC\setminus\bR$. It is not easy to find such Riemannian 
manifolds $M$ for which $J(i)$ exists. The only known examples are: certain
surfaces of revolution (\cite{Sz91}),  compact   normal 
Riemannian homogeneous spaces \cite{Sz3}, some examples obtained from the 
homogeneous examples using the method of symplectic reduction (\cite{Ag}).
 One can also take products or coverings of such spaces, but these are all
we know so far.
  
The original definition
 of adapted complex structures in \cite{LSz1} corresponds to the parameter 
$s=i$. Recall that $S$ denotes the upper half plane.
For $s\in S$, $(N,\omega,J(s))$ becomes a K\"ahler manifold, denote its
canonical bundle by 
 $K^s$.
  If $s=i$, we simply write $K$ instead of $K^i$.

\begin{prop}\label{P:31}
Assume $M$ is oriented. Then
the bundle $K^s$ is holomorphically trivial.
\end{prop}
This is Lemma 10.4.1 in \cite{LSz3}, but for the readers' sake we include its 
short  proof here.
\begin{proof}
 According to \cite[Theorem 6c]{LSz2}, 
the map $A_{\sigma_s}:(N,J(s))\ra(N,J(i))$ is a
 biholomorphism. Therefore
it is enough to deal with the case $s=i$. Since $(N,J(i))$ is a Stein manifold
(\cite[Theorem 5.6]{LSz1}), by the Oka principle (\cite[pp. 144-145]{Ho} , it
suffices to show that $K$ is smoothly trivial.
$M$ is a deformation retract in $N$, so we only need  to show that
$\left.K\right|_M$ is trivial.
Since $M$ is oriented, the bundle $K_M\ra M$ of real $m-$forms is trivial.
 But restricting a form in $\left.K\right|_M$ to $TM$ is an isomorphism
 $\left.K\right|_M\approx\bC\otimes K_M$ and we are done.

\end{proof}
 It is interesting, that although the canonical bundle $K$ is trivial,
the first guess for a trivializing holomorphic section is wrong. Namely
if one takes the (m,0) component of the volume form of $M$ and extend it
holomorphically to $N$ (assuming that such an extension exists not only in a
neighborhood of $M$ in $N$), this holomorphic (m,0)  form may have zeroes in 
$N\setminus M$, so it is not a trivializing section of $K$. This is the case
for example for those surface of revolution metrics on the two sphere in
\cite{Sz91}.

The formula
\begin{equation}\label{E:21}
h_s(\alpha,\alpha)\mathcal E(x)=i^{m^2}\alpha\wedge\bar\alpha\quad x\in N,\alpha
\in K^s_x
\end{equation}
defines a hermitian metric $h_s$ on $K^s$. Write $h$ for $h_i$.

Due to the proposition, when $M$ is orientable, there is a Hermitian holomorphic
(in fact trivial)
line bundle $(\kappa, h_{\kappa})$ so that $(\kappa\otimes\kappa,h_{\kappa}^2)\approx (K,h)$.
Let  $\Theta$ be a trivializing holomorphic section of $K$
 and $\theta$ the corresponding section of $\kappa$
with $\theta\otimes\theta=\Theta$. Taking $\kappa^s:=A_{\sigma_s}^*\kappa$,
$h_{\kappa_s}=A_{\sigma_s}^*h_{\kappa}$,
we have $(\kappa^s\otimes\kappa^s,h_{\kappa_s}^2)\approx (K^s,h_s)$.
Let
 $\Theta_s:=A^*_{\sigma_s}\Theta$ and $\theta_s:=A^*_{\sigma_s}\theta$.
\begin{prop}\label{P:32}
\[
h_s(\Theta_s,\Theta_s)=(\text{Im }s)^mh(\Theta,\Theta)\circ A_{\sigma_s}\quad
h_{\kappa_s}(\theta_s,\theta_s)^2=h_s(\Theta_s,\Theta_s)
\]
\end{prop}
\begin{proof}
Since $\chi(\sigma_s)=\text{Im }s$,
 \eqref{E:21} and \eqref{E:13} implies
\[
\Theta_s\wedge\bar\Theta_s=A^*_{\sigma_s}(\Theta\wedge\bar\Theta)=
A^*_{\sigma_s}(h(\Theta,\Theta)(-i)^{m^2}\mathcal E)=(-i)^{m^2}
(\text{Im }s)^m(h(\Theta,\Theta)\circ A_{\sigma_s})\mathcal E
\]

Using \eqref{E:21} once again we obtain the first formula from which the second follows from the definitions.
\end{proof}

\section{The field of prequantum Hilbert spaces}\label{S:fpreq}

The prequantum line bundle is a Hermitian line bundle $E\ra N$ with a hermitian
connection whose curvature is $-i\omega$. When $M$ is simply connected, it is
unique. In any case one such line bundle is the trivial line bundle
$E=N\times\bC\ra N$ with the trivial metric $h^E(x,\gamma)=|\gamma|^2$ on it.
The connection on $E$ is obtained from a real $1-$form $a$ on $N$ such that
$da=-\omega$ (such $a$ exists since $\omega$ is the pull back of the canonical symplectic form
of $T^*M$). Let $\vartheta$ be the canonical section of $E$ i.e. that is
defined by
$\vartheta(x):=(x,1)$,
$x\in N$.
Then all the sections of $E$ have the form $f\vartheta$, where $f:N\ra\bC$
and the connection on $E$
is defined by
\[
\nabla^E_\zeta(f\vartheta)=(\zeta f+ia(\zeta))\vartheta,
\]
where $\zeta$ is any smooth complex vector field on $N$.
The half-form corrected prequantum Hilbert space $H^{pr Q}_s$
corresponding to the
K\"ahler manifold $(N,\omega, J(s))$, $s\in S$ is the Hilbert space of $L^2$
sections of the bundle $E\otimes\kappa^s$. Since 
$\vartheta\otimes\theta_s$ is a
nowhere vanishing section of $E\otimes\kappa^s$,
\[
H^{pr Q}_s=\{\psi=f_s\vartheta\otimes\theta_s \mid f_s:N\ra\bC,
\int\limits_N|f_s|^2h_{\kappa_s}(\theta_s,\theta_s)\mathcal E<\infty
\}.
\]
Let
\[
H^{pr Q}:=\cup_{s\in S}^*H^{pr Q}_s
\]
be the disjoint union of these Hilbert spaces. $p:H^{pr Q}\ra S$ is the natural projection map.
At the moment there is no further structure defined yet on the set $H^{pr Q}$, except that the fibers $p^{-1}(s)$ are Hilbert spaces.
$p:H^{pr Q}\ra S$ is an example that we call (cf. \cite{LSz3}) a field of Hilbert
spaces.

\section{Smooth Hilbert bundle structures on $H^{pr Q}$}\label{S:smoothstr}

The  map 
\[
\begin{matrix} \mathcal L:=S\times L^2(N,\mathcal E)&\overset{A}\lra&H^{pr Q}\\
(s,f)&\lmt &\frac{f}{\sqrt{h_{\kappa_s}(\theta_s,\theta_s)}}\vartheta\otimes\theta_s
\end{matrix}
\]
is a  fiber preserving bijection and its restriction to each fiber is  
 unitary. Therefore pushing forward $\mathcal L$ with $A$,
 equips $H^{pr Q}$ with a smooth (in fact) trivial
Hilbert bundle structure. The canonical
flat hermitian connection on 
 $\mathcal L$
yields an orthogonal connection on $H^{pr Q}$.

When $M$ is a compact Lie group equipped with a biinvariant metric, this is
the Hilbert bundle structure with hermitian connection
chosen by C. Florentino, P. Matias,
J. Mour\~ao and J. Nunes in their papers \cite{FMMN1, FMMN2}, except that they do
not consider the full parameter space $S$, only the positive imaginary axes.

But this is not the only possible natural way to equip $H^{pr Q}$
with a Hilbert bundle structure.
Let $\psi_s=f_s\vartheta\otimes\theta_s\in H^{pr Q}_s$. 
Then from Proposition~\ref{P:32} we
get
\[
\|\psi_s\|^2_{L^2}=\int\limits_N|f_s|^2h_{\kappa_s}(\theta_s,\theta_s)\mathcal E=
(\text{Im }s)^{\frac{m}{2}}\int\limits_N|f_s|^2
(\sqrt{h(\Theta,\Theta)}\circ A_{\sigma_s})\mathcal E=
\]
\[
=(\text{Im }s)^{-\frac{m}{2}}\int\limits_NA_{\sigma_s}^*(|f_s|^2\circ A^{-1}_{\sigma_s}
\sqrt{h(\Theta,\Theta)}\mathcal E)=
(\text{Im }s)^{-\frac{m}{2}}\int\limits_N(|f_s|^2\circ A^{-1}_{\sigma_s})
\sqrt{h(\Theta,\Theta)}\mathcal E
\]


Therefore the map
\[
\begin{matrix} H^{pr Q}&\overset{B}\lra&\mathcal L=S\times L^2(N,\mathcal E)\\
f_s\vartheta\otimes\theta_s&
\lmt&(s,(\text{Im }s)^{-\frac{m}{4}}f_s\circ A^{-1}_{\sigma_s}
h(\Theta,\Theta)^{\frac1{4}})
\end{matrix}
\]
is a  fiber preserving bijection whose restriction to each fiber is  
 unitary.
By  pulling back the Hilbert bundle structure of $\mathcal L$ with $B$,
the Hilbert field $p:H^{pr Q}\ra S$ inherits
another smooth (in fact trivial) Hilbert bundle structure.
We claim that as a smooth bundle,
this is different from the one we obtained with the help
of the map $A$ earlier.

\begin{thm}\label{T:51}
The Hilbert bundle structures on $p:H^{pr Q}\ra S$ obtained
by the maps $A$ and $B$ are the same as topological Hilbert bundles but their
smooth structures are different.
\end{thm}
\begin{proof} Using Proposition~\ref{P:32} we
calculate the map $B\circ A:\mathcal L\ra \mathcal L$ to be
\[
(s,f)\mt(s, (\text{Im s})^{-\frac{m}{2}}
f\circ A^{-1}_{\sigma_s})=(s,\varrho((\sigma_s)^{-1})f). 
\]
It follows from Theorem~\ref{T:21}
that the fiberwise unitary map $B\circ A$ is a
homeomorphism but it is not differentiable.
\end{proof}

{\bf Acknowledgement} We thank L\'aszl\'o Lempert for stimulating discussions.

\end{document}